\documentclass[oneside]{article}

\usepackage{epsfig}
\usepackage{amssymb,color}
\usepackage{framed}
\usepackage{epsf}
\usepackage{algorithm, algorithmic}
\usepackage{booktabs}
\usepackage{multirow}
\usepackage{graphics,epsfig}
\usepackage[normalem]{ulem}
\usepackage{url}

\newcommand{\qed}{\mbox{}\hspace*{\fill}\nolinebreak\mbox{$\rule{0.6em}{0.6em}$}
}

\definecolor{gray}{rgb}{0.5,0.5,0.5}

\newtheorem{theorem}{Theorem}

\newtheorem{claim}[theorem]{Claim}
\newtheorem{corollary}[theorem]{Corollary}
\newtheorem{definition}[theorem]{Definition}
\newtheorem{remark}[theorem]{Remark}

\newtheorem{observation}[theorem]{Observation}
\newtheorem{statement}[theorem]{Statement}
\newenvironment{proof}{{\bf Proof:}}{$\qed$\par}

\newcommand{\fall}{$\forall$}
\newcommand{\implies}{$\rightarrow$}
\newcommand{\deduce}{$\models$}
 
\begin{document}

\title{Can we know the Whole Truth?}

%\numberofauthors{1}
\author{
Rina Panigrahy \\
Microsoft Research, Mountain View, CA \\
rina@microsoft.com
%\alignauthor
}

\maketitle

\newcommand{\Omg}{$\Omega$}

\begin{abstract}
The paper explores known results related to the problem of identifying if a given program terminates on all inputs -- this is a simple generalization of the halting problem. We will see how this problem is related and the notion of proof verifiers. We also see how verifying if a program is terminating involves reasoning through a tower of axiomatic theories -- such a tower of theories is  known as Turing progressions and was first studied by Alan Turing in the 1930's. We will see that this process has a natural connection to ordinal numbers. The paper is presented from the perspective of a non-expert in the field of logic and proof theory. 
\end{abstract}

\section{Introduction}

Science can be viewed as the quest for truth. Perhaps there is such a thing as the `whole truth' and only a small subset of it that we `know to be true' (of course we are ignoring some falsehoods that we think to be true).  Our goal is to increase the set of truths that we know to cover the whole truth; is this even possible? (the whole truth may even include deep philosophical truth but it may not be easy to agree on the definition of such truths.) For simplicity, we can focus on a set of concrete truths where the definition of truth is fairly unambiguous; for example, concrete questions such as `can one extract energy from algae efficiently' or `is there a program that quickly solves the travelling salesman problem'. Different scientific disciplines are interested in truths of different types: while a phycisist is interested in truths about matter and energy, a biololigist is more interested in truths  about humans and animals, and a computer scientist about the behavior of programs. Mathematics is interested in truths about abstract concepts such as numbers, geometric objects. It may seem that the space of abstract concepts is separate from the physical world but this is far from true. First, abstract concepts can be used to model things in the physical world and thus truths in the abstract world can inform us of truths in the physical world. Further certain objects such as programs can be viewed as both objects  in the physical world and as asbtract mathematical entities. Think of a computer program that controls the operation of a catrastrophic weapon -- the question of whether the program reaches a certain state can be formulated as a mathematical question. Thus theorem proving (finding truths in mathematics) should be viewed not just as a mental game of puzzle solving but as in fact it reflects our ability to discover truths in the physical world.

In this paper we will restrict ourselves to truths about programs  and further only to the question of whether a given program terminates or not. It is in this restricted subet of truths that we will investigate the limits of what we know and what we can ever know. In mathematics new truths (theorems) are discovered by the way of  proofs. If a truth has no proof we can never know it; if it has one it is can be known -- of course there is the matter of finding the proof so it is not the case that we can easily know that truth but if we search hard enough for the proof may be we will find it -- thus it is at least knowable.  Thus to understand what truths are knowable we need to understand what is a proof. In practice, proof checking is a psychological process and naturally an important question is whether we can understand  this process precisely.
 
The study leads one into a fascinating journey of (known) facts  in topics from Godel's theorems, ordinal numbers, Turing progressions, and philosophy of reasoning and knowledge. Perhaps much of the scientific community (except for the field of proof theory) Godel's theorem is often understood to mean that there will always be statements that are true that we will never know to be true -- this however is not really known to be true -- all that is known is that a 'mechanical verifier'(a program) can never accepts proofs for all true statements. From an epistemological perspective though what we can know from a given set of axioms (such as Peano Arithmetic) turns out to be harder to quantify -- its power is given by a fascinating and little understood phenomena called Turing progressions which can never be explicitly known to be enumerable. The significance of Turing progressions and its implications to the concept of a 'proof' may have been underplayed in the broad scientific community.  The paper is simply the  lessons learnt  by a non-expert in the field while exploring this question --  everything here is known already and is the result of discussions with several people of which at least one is an expert on the topic.

\section{The Problem and its difficulty}

Which programs terminate and which don't? This is the question we will study here. Specifically, we will be interested in checking if a program that takes inputs terminates on all inputs. This question turns out to be very hard to answer and its investigation reveals several astonishing facts. 

The set of terminating programs is known to be undecidable, non-enumerable (no Turing machine can list them all), and non-provable in an axiomatic system.  Here is what it means for the set of terminating programs to be decidable, enumerable and provable.

\newcommand{\noindentpara}{\vspace{2mm}\noindent}

\noindentpara{\bf Decidable:} There is a Turing machine that always halts and accepts a program if and only if it is terminating.

\noindentpara{\bf Enumerable:} There is a Turing machine that accepts a program if and only if it is terminating. It may reject or simply may not halt on other programs. Clearly to accept terminating programs it must halt at least on those inputs.

\noindentpara{\bf Provable (in an axiomatic system):} There is an axiomatic system in which it is possible to prove for each terminating program that it terminates. It should be possible to build an automated {\em verifier} that verifies proofs in this system. It should not prove a non-terminating machine to be terminating.

\noindentpara{\bf Provable (in practice):} It is possible to write a paper that proves that a certain program terminates. Reviewers should accept the proof. The proof should use as primitive assumptions as possible (such as mathematical induction). A conservative mathematician may prefer to avoid advanced assumptions such as  continuum hypothesis, or even say the axiom of choice in the proof.

Are the above two notions of provable the same? We will see that the answer is far from obvious. 
We will call the programs (Turing machines) of interest as terminating programs. Such a Turing machine T halts on all inputs x, written formally as ``$\forall$ x Halts(T,x)"
 We will say 
Terminating(T) := \fall x Halts(T,x).
 Let 
$\Omega$ = \{T: Terminating(T) \}. Thus \Omg\  is the set of all terminating programs.

\begin{theorem}
The set \Omg\ is not decidable and not enumerable. 
\end{theorem}

Given a finitely axiomatized theory, one can write a verifier (a program) that verifies proofs in that theory.  Thus given a theorem and  a proof as strings the verifier can check if it is a valid proof in that theory and if it results in the theorem which will usually be the last statement in the proof.  If a verifier accepts a proof of the statement ``Terminating(T)" then for convenience  we will say that it accepts the program T (it can do this non-deterministically by guessing the proof). We will say the verifier is correct if it does not accept any program outside \Omg. Thus for our purpose a verifier is defined as follows:

\begin{definition}
A {\em  verifier} is a program that takes as input a proof, turing machine pair (pi, T) and checks if pi correctly proves that T is a terminating program in  a certain axiomitic system. If it accepts (pi,T) then T must be in $\Omega$. The verifier must halt on all inputs (pi, T).
\end{definition}
 
 G\"{o}del's theorems~\cite{godel} can be used to infer some properties of the set of programs accepted by a verifier.
The following theorem is a direct consequence (rather a restatement) of G\"{o}del's theorems.
\begin{theorem}\label{godel}[G\"{o}del]

\begin{itemize}

\item In any finitely axiomatizable  proof system with a proof verifier, all machines in \Omg\ cannot be proved to be terminating unless it proves a machine outside \Omg\ to be terminating.

\item Further, given the code for a verifier V, one can explicitly write the code for another program P that is not proven to be terminating by V but is in \Omg.
\end{itemize}
\end{theorem}

In fact, given a verifier V, one can always write another verifier S(V) that is strictly stronger than V, in the sense that it accepts more terminating programs -- here S is a deterministic function that obtains one verifier from another.

In a philosophical essay ``Minds, Machines and G\"{o}del,"  Lucas
\cite{Lucas, Lucas2} suggested that G\"{o}del's theorems show that the human mind is more powerful than machines. This paper was heavily debated, heavily praised and criticized. Major supporters included Penrose ~\cite{Penrose, Penrose2, Penrose3}; some of the major critics were Hofstadter and Franz\'en\cite{Hof, Franzen3}.
Regardless of the points of the debate (see ~\cite{aaai, hypercomputing, chalmers, intuition, Franzen3} as a small sample of the arguments),  the following statements seem to be correct consequences of G\"{o}del's theorem (we refer to these as statements and not theorems as they talk about people.)

\begin{statement}[G\"{o}del]
If a person's reasoning can be modeled by a Turing machine, then there is some machine in \Omg\ that (s)he will not know to be terminating.
\end{statement}

Another way to think of G\"{o}del's theorem is as follows. Note that as stated before a verifier can be viewed as a Turing machine that accepts a subset of \Omg\  (if a verifier verifies proofs, then we will simply look at the subset of \Omg\ that can be proven to be terminating under that verifier).  A verifier  V is maximal in a set of verifiers if no other verifier accepts  a program to be terminating that is not already proven to be terminating by V.

\begin{theorem}\label{reasonable}[G\"{o}del]
\begin{itemize}

\item There is no maximal verifier. The set of verifiers is not enumerable.

\item Any reasonable belief system does not have a maximal trusted verifier.

\end{itemize}
\end{theorem}

Here a belief system is simply a collection of statements. A `trusted' verifier means that the belief system knows that whatever the verifier accepts is true. Informally, the above theorem says.

\begin{statement}[G\"{o}del]
A reasonable mathematician cannot point to a single Turing machine and claim that it enumerates all programs that (s)he believes to be terminating. For every Turing machine that enumerates terminating programs, (s)he can point to a better Turing machine that enumerates more terminating programs.
\end{statement}

This brings us to an important question: what process do mathematicians follow in accepting proofs? Is it a Turing machine? Is it certain that there are programs what we will never know to be terminating? Does there exist a terminating program for which we can never write a (correct) paper that proves it to be so?
To understand the complexity of these questions, we must first understand the essential ideas  (next section) behind G\"{o}del's theorems. In section~\ref{tprogressions} we will look at an interesting phenomena called Turing progressions that allows one to construct infinitely many (in fact we may even say {\em transfinitely} many) verifiers from a single verifier.  We will see how such progression of verifiers is related to {\em ordinal} numbers -- numbers beyond natural numbers. We summarize our conclusions about the questions posed above in section~\ref{summary}.

\section{Understanding G\"{o}dels argument}\label{godelproof}

We will now prove theorem~\ref{godel} (theorem~\ref{reasonable} follows from the same ideas and is proven in the appendix.)  The main idea is that given any correct verifier that accepts the set of programs $\Omega_0 \subset \Omega$, one can construct a terminating program not in $\Omega_0$;  and then one can obtain a better(stronger) verifier that accepts everything in $\Omega_0$ and the new terminating program.  The new program that one can show to be terminating is very simple. The proof uses diagonalization to construct a function that is different from each of the functions in an enumerable set of functions. We will work with programs that return a boolean value upon termination. Given a verifier V, a program T, and a proof pi that proves that T is terminating, we will use V(pi,T) to denote if V accepts the proof pi as a proof of the statement ``Terminating(T)". Thus V(pi,T) is also a boolean value. We will use the terms Turing machines and programs interchangeably.

%: given functions  $f_1(x), f_2(x),....f_i(x)...$ that map natural numbers to bits, one can construct a function $g(x)$ that 
%is different from all these. $g(x) = \not f_x (x)$. $g(x)$ is different from $f_i(x)$ as they differ on input $i$. 
This following program on input x, checks if x represents a proof, Turing machine pair ``(pi, T)" such that pi correctly proves to V that T is a terminating. If not it exits; otherwise, it
returns the opposite of T(x). We will prove that this program is outside the set of programs accepted by V to be terminating.

\begin{verbatim}
P = "function  P(x) 
       { 
          Let "($pi, $T)" = x  
         // i.e., parse the input x to check if it forms a  proof, Turing machine pair.
         // If so, store the Turing machine in the variable T and proof in pi
         // We use $T, $pi to denote the literal strings present in the 
         // variables T and pi respectively (in Perl style syntax)
         // If this is not possible return anything 

          if (V(pi,T))
          {
                // That is, check if the verifier accepts the Turing machine , proof pair.
                // The proof certifies the statement "\forall x Halts($T,x)" to V.
                // That is the  machine with code 
                // represented by string $T is terminating. 


              return not T(x)  // that is, return the opposite of T(x)
          }
          else 
               return 0. 
       }"
\end{verbatim} 

\begin{claim}
If V is a correct verifier, the program P always halts. Further it is not accepted by V (that is there is no proof of P to be terminating that is accepted by V)
\end{claim}
\begin{proof}
The claim is a standard diagonalization argument.

The program checks if x is a pair (pi,T) where T is a Turing machine and pi proves to V that T always halts.
Only if it passes the check, does it execute T(x). Thus the only reason it may not halt is that T does not halt. But T is proven
to be halting by V. So this cannot happen.

The above program is different from every program accepted by V. To see this note that if T is accepted by V, then P differs from T on
at least one input namely ``(pi,T)" where pi is the proof that T is a terminating program.
\end{proof}

\section{New Verifiers from Old: Turing Progressions}\label{tprogressions}
Given that one can obtain a program P from a verifier V, we will show how one can get a sequence of verifiers each better(stronger) than the previous one. One way to think of the verifier is that, instead of accepting valid proofs, it is a non-deterministic program that accepts theorems provable in the theory -- this is achieved by simply making the verifier ``guess" the proof of a given theorem it is checking for provability. Thus the (non deterministic) verifier accepts all theorems provable in the theory. Given such a verifier for a theory, since we are only interested in statements of the form  ``\fall x Halts(T,x)", we can modify it into a program that accepts T if and only if V (as a non deterministic program) accepts the statement ``\fall x Halts(T,x)". Thus the modified program accepts some subset  $\Omega_0 \subset \Omega$ assuming the original theory was correct. 

Given a (deterministic) verifier for proofs, one can (programmatically) convert it into a non-deterministic verifier that accepts provable statements and vice versa (from the non deterministic verifier, the proof is simply the set of non deterministic bits). As a deterministic verifier we will used the notation V(pi,T) to denote if V accepts pi as a proof of T to be a terminating program.  When we view the same verifier as a non-deterministic verifier we will use the notation V(T) to indicate if it accepts T as a terminating program -- it will non-deterministically guess the proof pi if one exists. In this section we will use the latter notation.

As shown in the previous section, given a verifier V, we can write a program P that is terminating but is not accepted by V.
But now that we know P is a terminating program we can construct a better verifier that accepts all machines accepted by V and also accepts the program P.
Note that even though the theory on which the initial verifier was based (say for example Peano Arithmetic)  did not ``know" that P is a terminating program, we do know this if we believe in the correctness of the theory. Thus if we were to write down
the above proof in a paper we would not expect a reviewer to object to the proof and the conclusion that P terminates.

This gives us a method to getter better verifiers from old ones. Let S(V) denote the better verifier.

\begin{verbatim}
S(V) = "
       function S(V) (x) 
       { 
            if (x== $P) 
                 return TRUE;
            else 
                return V(x);
       }"

\end{verbatim}

Thus the above verifier accepts everything that V accepts and also accepts the program P.  So if we use $L(V)$ to denote the language accepted by a verifier V, then  $\Omega_0 = L(V)$
and we know that if $\Omega_1 = L( S(V))$ accepts a larger subset  so $ \Omega_0 \subsetneq \Omega_1 $. Thus we get:

\begin{claim}
If V is a verifier, then S(V) is a strictly stronger verifier; that is, it accepts a larger subset of terminating programs than V.
\end{claim}

Instead of just adding one program P to the set of programs accepted by  V, we could do something smarter.  We could add an axiom to the axioms of V that will allow it to accept the new program P. One way to prove that the program P is terminating is to assume that whatever program V proves to be terminating is terminating. Usually a verifier accepts proofs based on some axioms and deduction rules -- we can put a wrapper around it so that it only accepts programs that can be proven to be terminating with those axioms.  So we could add the new axiom 
R = ``\fall x Accepts(V,x) \implies\ Terminating(x)". If this was available to a verifier as an axiom it can derive the statement ``\fall x Halts(P,x)" as long as it has  the basic axioms of arithmetic and axioms to prove properties of programs and the deduction rule of Modus Ponens. Modus ponens is the prime deduction rule in most axiom systems -- it says that if ``x" and ``x \implies\ y" can be derived then ``y" can also be derived. So using modus ponens, from R, one can derive the statement ``\fall x Halts(T,x)" for any T that is known t be terminating. Hence one can formally deduce that P is also always terminating. The statement R denotes what is called as 1-consistency in proof theory.
In this case we will  implement S as follows: S(V) = Verifier obtained by adding the statement ``\fall x Accepts(V,x) \implies\ Terminating(x)" as  new axiom to the axioms of V. S can be formally written as a program just like the previous one.

\newcommand{\union}{\cup}

Thus S can be viewed as a function that takes a function (verifier) as an input and outputs a new function. Starting from $V_0 = V$, we can get $V_1 = S(V)$, $V_2 = S(V_1)$,...,$V_i = S^i (V)$. Even though the each $V_i$ accepts more
terminating machines than $V_0$, the above arguments show that they all accept only terminating programs. In fact for any fixed $i$, one can write down an explicit proof of its correctness in a way that would be acceptable to any mathematician who believes in the correctness of the base verifier (if you ask:  in which axiomatic system is the proof written, the response would be: think like a person reviewing the proof not like a Turing machine.)
But we can even write a single verifier $V_B$ (where B is used a symbol for a very BIG number greater than every natural number) that accepts the union of all these verifiers.
That is $L(V_B) = \union_{i \in N} L(V_i)$
 
The code for $V_B$ looks like this:

\begin{verbatim}
V_B = "     
 
      function V_B (x) 
       { 
            i = some natural number chosen non-deterministically;

            return V(i)(x)
       }

      function V(i) {
            if (i==0) return V;
            else return S(V(i-1))
      }

 }"
\end{verbatim}

\newcommand{\Caret}{$\wedge$}. 

But then now can define  $S(V_B), S(S(V_B)),...., S^i(V_B)$. Let's denote these by $V_{B+i}$. One can write a program for each of these verifiers. Not only that we can actually write a program that enumerates all these verifiers and accepts all 
their proofs giving us the new verifier  $V_{B+B}$ or let's call it $V_{B.2}$. Then we can get $V_{B.3},..,V_{B.i}$. Finally we can enumerate all these and write a single verifier $V_{B.B}$ that accepts their union. Continuing this way we get

\[V_0,V_1,...., V_B, V_{B+1}, V_{B+2},... V_{B.2},V_{B.3},V_{B.4},V_{B.i},..
\]
\[
V_{B.B}..,V_{B.B.B},..,V_{B^i},....V_{B^B},V_{B^{B^B}},....
V_{B^{B^{B^{..i\ times}}}},...V_{B^{B^{B^{.. B\ times}}}} .....
\]
Note that $V_B$ can be viewed as a verifier enumerator as it essentially enumerates all previous verifiers. $V_{B.(i+1)}$ is obtained from $V_{B.i}$ by writing a higher level function that translates one verifier enumerator into another. Similarly we get $V_{B.B}$ by writing a program for a verifier enumerator enumerator. $V_{B^i}$ is a verifier enumerator to the ith degree. $V_{B^B}$ enumerates all verifier enumerators of any finite degree.
Anyone with some familiarity with ordinal numbers will see the pattern here. This progression of theories is called Turing progressions. Alan Turing looked at such progressions in some papers 1937-1940  \cite{turing}. This was later taken up by 
Solomon Feferman \cite{fefferman, fefferman2} around 1960 and logicians Beklemishev \cite {Bek} and Artemov \cite{Artemov}  more recently. The main principle used in formalizing  such a progression is to assume some notion of the consistency (called as 1-consistency) of the previous theory in the next theory. By G\"{o}del's theorem the consistency  of the theory given by a verifier V (which is akin to the statement ``\fall x Accepts(V,x) \implies\ Terminating(x)") is not provable by V.  Assuming this consistency (which doesn't really seem like an assumption) is called as the {\em reflection principle} in proof theory.

\noindentpara {\bf Ordinal numbers:}
Ordinal numbers \cite{WikiOrd, cantor} are extensions of natural numbers introduced by Cantor to talk about  infinite sequences 
The finite ordinals are the natural numbers: 0, 1, 2, …,. Each number in this sequence can be viewed as the size of the set of all the numbers before it. Next, we need a number that denotes the size of the set of all the natural numbers and we denote this by  the least infinite ordinal  $\omega$. Next we have $\omega+1$,..,$\omega+2$,.... Next we get $\omega . 2$ which is the least ordinal that is greater than all ordinals of the form $\omega+i$. Continuing this way we get

\[
1,2,...,\omega, \omega + 1, \omega + 2, ..., \omega . 2, \omega . 2 + 1, ..., \omega^2, ..., \omega^3, ..., \omega^\omega, ...,\omega^{\omega^\omega}, ..., \omega^{\omega^{\omega^{{.}^{.}}}}, ...
\]

At some point you will get an ordinal number that represents the size of an uncountable set. We will say an ordinal is countable if the set of ordinals less than it (in the sequence) is countable. 
Some ordinals are simply obtained by adding $1$ to the previous ordinals; others such as  $\omega$ are called limit ordinals -- one cannot really point to a unique ordinal just before $\omega$.  How should we denote or visualize ordinals beyond the ones written above? This in itself is a very difficult question and has ramifications into questions related to Turing progressions. 
For our purposes, we can view the code of the verifiers that we get as a natural {\em notation} for some of the initial ordinal numbers ( this idea was suggested by Fefferman ). For a limit ordinal, as long as we can explicitly {\em enumerate} the verifiers before that point, we can write the code for the verifier corresponding to that ordinal.

Can we get verifiers for every countable ordinal? Apparently not; the highest ordinal for which we hope to get a verifier is bounded by something called the Church-Kleene ordinal~\cite{WikiOrd} -- ordinals less than this are called computable ordinals.  It is known how to get notations and corresponding verifiers up to some ordinal in this sequence way way below the Church-Kleene ordinal. But beyond that the questions are open. The more notations or verifiers we find the larger the subset of \Omg\ we know to be terminating. Now, it is known that {\em if} we had a good notation or verifier for every computable ordinal, we could in fact cover all of \Omg\ -- that is we would be able to prove all the terminating programs (Turing-Fefferman Completeness theorems~\cite{fefferman}). But unfortunately we do not know how to express the higher ordinals beyond a certain point well. Is the (good) notation of ordinals beyond a certain point known to be impossible? No. But Beklemishev says it seems difficult and may even be unlikely.

Another interesting question is whether it is possible to get richer theories from simpler ones by applying the reflection principle repeatedly; for example can one get ZFC (Zermelo Fraenkel set theory) from PA (Peano Arithmetic) at a certain ordinal? This question is studied by Beklemishev in \cite{Bek}. It turns out that using natural orinal notations,  one can go very far constructing the tower without ever reaching ZFC.  We get at ZFC at an ordinal that is called as the  ``proof-theoretic $\Pi_2^0$-ordinal of ZFC", but it is not known how big it is (there are only some known lower bounds). It is considered, perhaps, the most important open question in classical proof-theory.  However, people have analysed and computed similar ordinals exactly for some other, rather strong, axiomatic systems. There is a whole science of computing such ordinals called ``proof-theoretic ordinal analysis" (see ~\cite{Bek} and the references there in) .

For further exploration of this rich and mysterious phenomena of Turing progressions the reader is directed to \cite{fefferman, fefferman2, Franzen, Franzen2, Bek, Bek2} (warning: if you are not familiar with proof theory, reading the literature may a daunting task).

\section{Summary}\label{summary}

Here are some of the conclusions we can draw from the arguments presented so far.

\begin{itemize}
\item Is there a program that decides if a certain program is terminating? No.

\item Is there a single program that for every terminating program can accept a proof that it terminates? No.

\item For every verifier that accepts proofs for some subset of terminating programs to be terminating, can one always write a better more powerful verifier to do so? Yes.

\item Can people provably and correctly know any given non-terminating programs to be non-terminating? We don't know.

\item Can people write a program and claim that the program enumerates all terminating programs that people know to be terminating? No.

\item Are people more powerful than programs when it comes to knowing if a program is terminating or not? We don't know.

\item Do we understand the process by which people accept proofs of a program to be terminating? Don't think so. But we do seem to be able to {\em execute} the process.

\end{itemize}

\subsection*{Acknowledgements}
The paper resulted out of my wanderings into questions about provably halting programs and owes much to discussions with several people. I am grateful to Lev Beklemishev for explaining known facts about Turing progressions and for reviewing a draft of this paper -- much of the details about Turing progressions presented here came directly from those discussions. I am also grateful to Sergey Yekhanin who gave a patient ear to my myriad confused arguments related to this topic. I would also like to thank Herman Ruge Jervell for telling me that the progression of theories is called Turing progressions and for pointing me to Lev's work. Finally I would like to thank Toni Pitassi, Mark Braverman, Ran Raz, Gil Segev, and Alex Andoni for useful and interesting discussions  about a certain `series of proof verifiers' that seemed to come from one verifier -- in hindsight this simply turned out to be the Turing progressions.

\appendix

\section{Proof of theorem~\ref{reasonable}}
The first part of the theorem follows from the fact that for every verifier one can construct a better verifier. Note that the set of verifiers cannot be enumerable as otherwise there is a single verifier that can enumerate them all and capture their combined power. We now formalize the second part and prove it.

Let a belief system B denote a subset of statements in some language. These are interpreted to be the set of statements that are believed to be true by a person.
V is a verifier in B if it accepts programs and for every program P accepted by V, the statement ``Terminating(P)"  is in B.
V is a trusted verifier if the statement ``\fall x Accepts(V,x) \implies\ Terminating(x)" is in B.
Let Trusted(V) := \fall x Accepts(V,x) \implies\ Terminating(x)

Assume all statements have balanced parenthesis and may be removed if there is no ambiguity.
A belief system is reasonable if it is closed under the following axioms/deduction rules:

\begin{itemize}
\item[1.] Modus Ponens (MP) 

\item[2.]  Trusted(\$V)  \implies\ Terminating(\$P) where P is code for the program defined earlier that depends on the code from V.

This says that if V is a trusted verifier the program P is also believed to be terminating.

\item[3.]  Terminating(\$P) \implies\ Trusted(\$V) where V = ``V(x)\{ if x=``\$P" return TRUE else FALSE\}"

This says that if a program P is believed to be terminating then a verifier that only accepts P is a trusted verifier.
\end{itemize}

\begin{theorem}
 A reasonable and correct belief system  doesn't have a maximal trusted verifier
\end{theorem}

\begin{proof}
 Consider any trusted verifier V; we will show that it cannot be maximal.  Let T(V) be the set of programs T accepted by the verifier V (to be terminating.)  Look at the program P. It is different from every machine in T(V). So P is not in T(V). But since V is trusted, from axiom 2 (and MP), we know that statement  S $\in$  B. Also by axiom 3, (and MP) there is a trusted verifier that accepts S. So V cannot be maximal.
\end{proof}

\if 0
Say a Turing machine T is an verifier-producer if  every string accepted by T is a trusted verifier.
We will say it is a trusted verifier producer if "\fall y Accepts(T,y) \implies\ (\fall x Accepts(y,x) \implies\ x)"  is in B.

\begin{theorem}
 For any reasonable belief system B there is no trusted-verifier-enumerator that enumerates all trusted-verifiers.
\end{theorem}
\begin{proof}
For every subset that is enumberable there is an element that is greater?????
\end{proof}
\fi


\begin{thebibliography}{99}



\bibitem{Artemov}
S. N. Artemov, On Explicit Reflection in Theorem Proving and Formal Verification 
, Lecture Notes in Computer Science, 1999, Volume 1632,1999, 678, DOI: 10.1007,3-540-48660-7 23 


\bibitem{Bek}
L. Beklemishev. Proof-theoretic analysis by iterated reflection. Archive for Mathematical Logic 42 (6), 515-552 (2003).

\bibitem{Bek2}
L. Beklemishev,   Induction rules, reflection principles, and provably recursive functions. Annals of Pure and Applied Logic, 85, 193-242. (1997)

\bibitem{Boolos1}
G.S. Boolos, Introductory note to: some basic theorems on the foundations of mathematics and their
implications, in: Collected Works III, Oxford University Press, Oxford, UK, 1995, pp. 290–304.

\bibitem{Boolos2}
G.S. Boolos, R.C. JeMrey, Computability and Logic, Cambridge University Press, Cambridge, UK, 1989

 \bibitem{cantor}
   Cantor, G., (1897), Beitrage zur Begrundung der transfiniten Mengenlehre. II (tr.: Contributions to the Founding of the Theory of Transfinite Numbers II), Mathematische Annalen 49, 207-246

\bibitem{chalmers}
D Chalmers, Minds, machines, and mathematics - Psyche, 1995

\bibitem{hypercomputing}
S Bringsjord, K Arkoudas, The modal argument for hypercomputing minds - Theoretical Computer Science, 2004 - Elsevier

\bibitem{fefferman}
S Feferman, Transfinite recursive progressions of axiomatic theories- Journal of Symbolic Logic, 1962

\bibitem{fefferman2}
S Feferman, Reflecting on incompleteness - Journal of Symbolic Logic, 1991 

\bibitem{fefferman3}
S Feferman, Three conceptual problems that bug me, Unpublished lecture text for 7th Scandinavian Logic Symposium, Uppsala, 1996. Available at \url{http://math.stanford.edu/~feferman/papers/conceptualprobs.pdf}.

\bibitem{feffermanSpector}
S Feferman, C Spector, Incompleteness along paths in progressions of theories - Journal of Symbolic Logic, 1962 

\bibitem{Franzen}
T. Franz\'{e}n. Inexhaustibility: A Non-Exhaustive Treatment.Wellesley ,Massachusetts

\bibitem{Franzen2}
T. Franz\'{e}n, Transfinite progressions: a second look at completeness - Bulletin of Symbolic Logic, 2004

\bibitem{Franzen3}
Torkel Franz\'{e}n. G\"{o}del's Theorem: An Incomplete Guide to its Use and Abuse. Wellesley, Mass.: AK Peters, 2005. x+ 172 pp. ISBN1− 56881− 238− 8.

\bibitem{Fenstad}
J.E. Fenstad, On the completeness of some transfinite recursive progressions of axiomatic theories- Journal of Symbolic Logic, 1968

\bibitem{godel}
G\"{o}dels Incompleteness Theorems. \url{http://en.wikipedia.org/wiki/G%C3%B6del's_incompleteness_theorems}

\bibitem{Hof}
D. Hofstadter, Waking up from the Boolean dream, in: Metamagical Themas: Questing for the Essence
of Mind and Matter, Bantam, New York, NY, 1985, pp. 661–665.

\bibitem{aaai}
M. Kerber, Why Is the Lucas-Penrose Argument Invalid? - KI 2005: Advances in Artificial Intelligence, 2005

\bibitem{Kugel}
P. Kugel, Thinking may be more than computing, Cognition 18 (1986) 128–149.


\bibitem{papadimi}
H. Lewis, C. Papadimitriou, Elements of the Theory of Computation, Prentice-Hall, Englewood CliMs,
NJ, 1981.


\bibitem{Lucas}
J.R. Lucas, Minds, machines, and G\"odel, in: A.R. Anderson (Ed.), Minds and Machines,
Prentice-Hall, Englewood CliMs, NJ, 1967, pp. 43–59 (Lucas’ paper is available online at
\url{http://users.ox.ac.uk/∼jrlucas/mmg.html}).

\bibitem{Lucas2}
J.R. Lucas, Minds, machines, and G7odel: a retrospect, in: P. Millican, A. Clark (Eds.), Machines and
Thought: The Legacy of Alan Turing, Oxford University Press, Oxford, UK, 1996, pp. 103–124.

\bibitem{Minsky}
Minsky, Marvin L.,
 {Computation: finite and infinite machines},
 1967,
 ISBN 0-13-165563-9,
 {Prentice-Hall, Inc.},
 {Upper Saddle River, NJ, USA}
 
\bibitem{WikiOrd} Ordinal Numbers, \url{http://en.wikipedia.org/wiki/Ordinal_number}.
Large Countable Ordinals. \url{http://en.wikipedia.org/wiki/Large_countable_ordinal}

\bibitem{Penrose}
R. Penrose, The Emperor’s New Mind, Oxford University Press, Oxford, UK, 1989.

\bibitem{Penrose2}
 R. Penrose, Shadows of the Mind, Oxford University Press, Oxford, UK, 1994.

\bibitem{Penrose3}
R. Penrose, Beyond the doubting of a shadow: a reply to commentaries on Shadows of the Mind,
Psyche 2 (3) (1996), This is an electronic publication. It is available at http://psyche.cs.monash.edu.au/
v2/psyche-2-23-penrose.html.

\bibitem{turing}
A. M. Turing 
Systems of Logic Based on Ordinals
Proc. London Math. Soc. 1939 s2-45: 161-228. Available online at \url{http://www.turingarchive.org}.

\bibitem{intuition}
C Wright, Intuitionists are not (Turing) machines - Philosophia Mathematica, 1995 


\end{thebibliography}
\end{document}